\documentclass[a4paper,12pt,reqno]{amsart}
\usepackage[PetersLenny]{fncychap}
\usepackage[latin1]{inputenc} 
\usepackage[T1]{fontenc} \usepackage[english ]{babel,minitoc}
\usepackage[nice]{nicefrac} 
\usepackage{graphicx,fancyhdr,fancybox,rotating,xcolor,appendix}
\usepackage{ulem,hhline,dsfont,pifont,multicol,lmodern}
\usepackage{amsmath,amsthm,amsfonts,amsbsy,amssymb,geometry,times,pst-node}
\usepackage{mathpazo} \usepackage[pdftex]{hyperref}
\numberwithin{equation}{section}  \makeatletter\@addtoreset{equation}{section}
\usepackage{euscript}
\usepackage{textcomp} \usepackage{sidecap} \usepackage{pdfpages} \usepackage{setspace}
\newcommand{\norm}[1]{\left\Vert#1\right\Vert} 
\newtheorem{theorem}{Theorem}[section]
\newtheorem{proposition}[theorem]{Proposition}

\newtheorem{lemma}[theorem]{Lemma}

\newtheorem{remark}[theorem]{Remark}
 
\newcommand{\C}{\mathbb{C}}

\newcommand{\bz}{\overline{z}}	
\usepackage[framemethod=tikz]{mdframed}
\usepackage{lipsum}
\definecolor{mycolor}{rgb}{0.122, 0.835, 0.998}
\newmdenv[innerlinewidth=0.5pt, roundcorner=4pt,linecolor=mycolor,innerleftmargin=6pt,
innerrightmargin=6pt,innertopmargin=6pt,innerbottommargin=6pt]{mybox}
\begin{document}
	
	\title[]{Integral representation for Jacobi polynomials and application to heat kernel on quantized sphere}
	
		\author{Ali Hafoud}
	\address{(A.H.) Centre Régional des Métiers de l'Education et de la Formation\newline de kenitra, Morocco} 
	\email{hafoudaliali@gmail.com}
	
		\author{Allal Ghanmi}   
	\address{(G.A.) Analysis, P.D.E. $\&$ Spectral Geometry, Lab MIA-SI, CeReMAR
		\newline Department of Mathematics, Faculty of Sciences, P.O. Box 1014
		\newline Mohammed V University in Rabat, Morocco}
	\email{allalghanmi@um5.ac.ma}

	\date{\today}
	\maketitle
	
	\begin{abstract}
	We derive a novel integral representations of Jacobi polynomials in terms of the Gauss hypergeometric function. Such representation is then used to give the explicit integral representation for the Heat kernel on the quantized Riemann sphere. 
\end{abstract}
	


 \section{Introduction} 
Integral representation of orthogonal polynomials have potential applications in several branches of mathematical, physical, statistical and engineering sciences,
see e.g. \cite{Askey1975,Erdelyi1953,Rainville71,Whittaker-Watson1952}. The following one \cite[Theorem 2.2]{DijksmaKoornwinder1971},
\begin{align}\label{IntRepJPol}
	P_n^{(\alpha, \beta)}(1-2t^2)=    c_{\alpha,\beta}^n
	\int_0 ^1 C_{2n}^{(\alpha+\beta+1)} (tu)(1-u^2)^{\alpha - \frac 12 } du,
\end{align}
is well-known ones for Jacobi polynomials. Above 
$$ 
c_{\alpha,\beta}^n :=  \frac{ 2 (-1)^n \Gamma( \alpha+\beta+1)  \Gamma(n+\alpha+1)}{\sqrt{\pi} \Gamma(n+\alpha+\beta+1) \Gamma \left( \alpha+ \frac 12 \right) }.
$$

 In the present paper we provide in Section 2 new integral representations for Jacobi polynomials such as the one involving the product of the Gauss hypergeometric function $_2F_1$ and the Gegenbauer polynomials. Namely we prove
 \begin{align}
 P_\ell^{(n, m)}(\cos(2 \theta)) &= \frac{ 2n!(\ell+m)!}{\pi (\ell+n+m)!}  \frac{1}{\cos ^{m}(\theta)} \int_{\theta}^{{\pi}/2} \frac{\sin(u)} { \sqrt  {\cos^2(\theta)-\cos^2(u)}}
 \\& \qquad\qquad \times 
 {_2 F_1}\left( 
 \begin{array}{c}
 -m , m \\ \frac 12
 \end{array}  \Big| \frac{\cos (\theta) -\cos (u)}{2 \cos (\theta) }\right)    C_{2l+m}^{n+1}(\cos u) du. \nonumber
 \end{align}
 As immediate application, we give in Section 3 an explicit integral representation of the Heat kernel for the invariant magnetic Laplacian
 \begin{align} \label{MagnLap}
 \Delta_{\nu} =
 -(1+|z|^2)^2 \frac{\partial^ 2}{\partial z \partial \overline{z}} -\nu (1+|z|^2)\left(  z\frac{\partial}{\partial z }-\overline{z}\frac{\partial}{\overline{z}} \right) +\nu^2 |z|^2
 \end{align}
 acting on the sections of the $ U(1)$-bundle for the ( quantized) Riemann unit sphere $ S^2$ identified to the extended complex plane $\mathbb{C} \cup \infty $,  
 and describing the Dirac monopole with charge $ q=2\nu$; $\nu>0$, under the action of a constant quantized magnetic field of strength $ \nu \in \mathbb{Z}^+$. For complement, we also provide in Section 4 a new direct proof of \eqref{IntRepJPol} which tied up to Dirichlet--Mehler integral and the Christoffel--Darboux formula for Jacobi polynomials.

\section{New integral representations of Jacobi polynomials}

We begin with the following result which readily follows by specifying $y=1$ in the Christoffel--Darboux formula for Jacobi polynomials \cite[Theorem 3.2.2, p. 43]{Szego1945}
and next making use of the three terms recurrence  formula for Jacobi polynomials in \cite[Eq. (2.17), p. 9]{Askey1975} (see also \cite[Chap. 4]{Szego1945}  or \cite[Chap. 10]{Erdelyi1953}).

  \begin{lemma}\label{lemFundamentalRessult1} 
  	The following formula
       \begin{align}
        \sum_{k=0}^{\ell} (2k+\alpha+\beta+1)\frac{\Gamma(k+\alpha+\beta +1)}{\Gamma(k+\beta+1)} P_k^{(\alpha, \beta)}(x)= \frac{\Gamma(\ell+\alpha+\beta +2)}{\Gamma(\ell+\beta +1)}P_\ell^{(\alpha +1, \beta)}(x) 
        \end{align}
        holds true for every $ \alpha > -1/2$, $ \beta >  -1/2$ and $-1\leq x<1$.
  \end{lemma}

Using Lemma \ref{lemFundamentalRessult1}, the Dirichlet--Mehler integral \eqref{IntRepDirichlet--Mehler} for Legendre polynomials \cite{Fejer1908}
(see also \cite[Eq. (3.1), p. 19]{Askey1975}) 
\begin{align}\label{IntRepDirichlet--Mehler}
P_{\ell}(\cos(2\theta ))=\frac{2}{\pi}\int_{\theta}^{{\pi}/2} \frac{\sin((2\ell+1)u)}{\sqrt{\cos^2(\theta)-\cos^2(u)}}  du
\end{align}
 as well as the observation
\begin{align}\label{sinsin} \frac{\lambda \sin (\lambda u)}{\sin(u) } = \frac{-1}{\sin(u)}\frac{d }{du}\left(\cos(\lambda u) \right),
\end{align}
we can prove the following

  \begin{proposition}\label{propFundamentalRessult2} 
  	For every nonnegative integers $n,\ell$ we have
  	\begin{align}\label{i}
  	P_\ell^{(n, 0)}(\cos(2 \theta) )
  	&= \frac{2 \ell!}{\pi 2^{n} (\ell+n)!} \int_{\theta}^{{\pi}/2} 
  	\frac{\sin(u)}{\sqrt{\cos^2(\theta)-\cos^2(u)}}\\
  	&\qquad \times
  	\left( \frac{-d}{\sin(u) du}\right) ^n \left( \frac{\sin((2\ell+n+1)u)}{\sin(u)}\right) du.\nonumber
  	\end{align}
\end{proposition}

 \begin{proof}  
	The proof of \eqref{i} follows by mathematical induction on $n$. The case of $n=0$ is exactly the Dirichlet--Mehler integral \eqref{IntRepDirichlet--Mehler} for Legendre polynomials. 
	Next, assume that  \eqref{i} for $P_k^{(n, 0)}(\cos 2 \theta )$ holds true for given fixed positive integer $n$ and all nonnegative integer $k$. 
	Therefore, making use  of Lemma \ref{lemFundamentalRessult1} we get
	\begin{align*}
	\frac{(\ell+n+1)!}{\ell !}  P_\ell^{(n +1, 0)}(\cos(2\theta))
	=\sum_{k=0}^{\ell} (2k+n+1)\frac{(k+n )!}{k!} P_k^{(n, 0)}(\cos(2\theta)).
	\end{align*}
	Hence, by induction hypothesis combined with the observation
	\begin{eqnarray}\label{sinsin} 
	\frac{\lambda \sin (\lambda u)}{\sin(u) } = \frac{-1}{\sin(u)}\frac{d }{du}\left(\cos(\lambda u) \right)  ,
	\end{eqnarray}
	we obtain
	\begin{eqnarray} \label{Recn1}
	\frac{(\ell+n+1)!}{\ell !}  &&P_\ell^{(n +1, 0)}(\cos(2\theta))
	\\\qquad & =& 
	\frac{1}{2^{n-1}\pi } \int_{\theta}^{{\pi}/2} 
	\frac{\sin(u)}{\sqrt{\cos^2(\theta)-\cos^2(u)}}
	\left( \frac{-d}{\sin(u) du}\right)^{n+1} \left(S_{\ell,n}(u) \right) du ,\nonumber
	\end{eqnarray} 
	with
	\begin{align*}
	S_{\ell,z}(u) &:= \sum_{k=0}^{\ell}   \cos((2k+z)u)
	= \frac{1}{2} \left(  \frac{\sin (z-1)u}{\sin(u)} + \frac{\sin ((2\ell+z+1)u)}{\sin(u)}\right) 	 
	\end{align*}
	which readily follows by direct computation. 
	Therefore, by taking $z=n+1$ and using the fact 
	\begin{eqnarray}\label{idsin} 
	\left( \frac{-d}{\sin(u) du}\right) ^{n}\left( \frac{\sin (nu)}{\sin(u)} \right)  =0 ,
	\end{eqnarray}
	we obtain 
	\begin{eqnarray}\label{actionDer} \left(  \frac{-d}{\sin(u) du}\right) ^{n+1} \left( S_{\ell,n+1}(u) \right) 	= \frac{1}{2} 
	\left(  \frac{-d}{\sin(u) du}\right) ^{n+1} \left(\frac{\sin ((2\ell+n+2)u)}{\sin(u)}\right).
	\end{eqnarray}
	Substitution of \eqref{actionDer}  in \eqref{Recn1} shows that \eqref{i} holds true for rank $n+1$ and for every nonnegative integer $\ell$. 
	This finishes the proof of Lemma \ref{propFundamentalRessult2}.
\end{proof}

\begin{remark} The identity \eqref{idsin} is immediate for ${\sin (nu)}/{\sin(u)}$
	being a ultraspherical polynomial in $\cos(u)$ of degree $n-1$ (see  \eqref{sinGrn}). 
\end{remark}

The previous result can be rewritten in terms of ultraspherical polynomials using  \eqref{sinsin}  as well as the well-known fact \cite[p. 218]{MagnusOberhettingerSoni1966}
\begin{align}\label{sinGrn}
\frac{\sin (nu)}{\sin(u)}= C_{n-1}^{(1)}(\cos u); \quad n=1,2, \cdots.
\end{align}

\begin{lemma}\label{lemIntJacGegn0}
	For every nonnegative integers $n,\ell$ we have
		\begin{align}\label{RelJacGegen} 
	P_{\ell}^{(n, 0)}(2t^2-1)=   \frac{2\ell!n! }{  \pi (l+n)!} \int_0^{1} \frac{C_{2\ell}^{(n+1)}(tv)}{ \sqrt{1-v^2} }  dv.
	\end{align}
\end{lemma}

 \begin{proof} 
	Recall first that  the ultraspherical polynomials satisfy
	\begin{eqnarray} \label{dnGegen} \frac{d^n}{dx^n}C_{\ell+n}^{(\lambda) } (x)=   \frac{2^n\Gamma(\lambda+n)}{\Gamma(\lambda)} C_{\ell}^{(\lambda+n) } (x) .
	\end{eqnarray}
	This can be handled by induction starting from $ \frac{d}{dx}C_{\ell+1}^{(\lambda) } =  2 \lambda C_{\ell}^{(\lambda+1) }  $.
	Then when combined with \eqref{sinsin} and the identity \eqref{sinGrn}, it infers
	\begin{align*}  
	\left(  \frac{-d}{\sin(u) du}\right) ^{n} 	\left(  \frac{\sin ((2\ell+n+1)u)}{\sin(u)} \right)
	&	=	\left(  \frac{-d}{\sin(u) du}\right) ^{n} \left(  C_{2\ell+n}^{(1)} (\cos (u)) \right) 
	\\&	=     2^{n} n!C_{2\ell}^{(n+1)} (\cos (u)).
	\end{align*} 	
	Therefore, from \eqref{i} one obtains \eqref{RelJacGegen}
	by means of the changes $ t=\cos (\theta) $ and 
	$v=\cos(u)/t$. This completes the proof.
\end{proof}

\begin{remark} The identity \eqref{RelJacGegen} appears as particular case of DijksmaKoornwinder integral representation of Jacobi polynomials given through \eqref{IntRepJPol}.  However, th \eqref{RelJacGegen} 
	can be use to reprove \eqref{IntRepJPol} making use of  Dirichlet--Mehler integral \eqref{IntRepDirichlet--Mehler} for the Legendre polynomials. Namely, we claim
	 we have
		\begin{align}\label{ii}
		P_\ell^{(n, m)}(2t^2-1 )= d_{n,m}(\ell)  \int_0 ^1 \left( 1-v^2\right)  ^{m-\frac 12} C_{2\ell}^{(n+m+1)}(vt) dv ,
		\end{align}
		For every nonnegative integers $m,n,\ell$ such that $ n \geq m$, 
		where 
		\begin{align}\label{CSTlEM}
		d_{n,m}(\ell) =:
		\frac{2^{2m+1} (\ell+m)!m!(n+m)!}{\pi (2m)!(\ell+n+m)!} .
		\end{align}
\end{remark}

Now, using the hypergeometric representation of ultraspherical polynomials, 
\begin{eqnarray}\label{GegenGauss} 
C_{2l}^{(\lambda)}(t) = (-1)^\ell \frac{\Gamma(\lambda+\ell)}{\ell! \Gamma(\lambda) } 
{_2 F_1}\left( 
\begin{array}{c}
-\ell , \ell+\lambda \\ \frac 12
\end{array}  \Big| t^2\right) ,
\end{eqnarray}
we can rewrite  \eqref{ii} in terms of the Gauss hypergeometric function  
\begin{align*}
P_{\ell}^{(n,m)}(2t^2-1) 
=\frac{2 (-1)^\ell (m+\ell)!}{\sqrt{\pi} \ell!   \Gamma\left(m+\frac 12 \right) } \int_0^{1} (1-v^2)^{m-1/2} {_2 F_1}\left( 
\begin{array}{c}
-\ell , \ell+n+m+1 \\ \frac 12
\end{array}  \Big| t^2v^2\right)  dv.
\end{align*}
Moreover, we can prove the following

\begin{theorem}\label{corIntRepJPGauss}
	We have
	\begin{align}
	P_\ell^{(n, m)}(\cos(2 \theta)) &= \frac{ 2n!(\ell+m)!}{\pi (\ell+n+m)!}  \frac{1}{\cos ^{m}(\theta)} \int_{\theta}^{{\pi}/2} \frac{\sin(u)} { \sqrt  {\cos^2(\theta)-\cos^2(u)}}
	\\& \qquad\qquad \times 
	{_2 F_1}\left( 
	\begin{array}{c}
	-m , m \\ \frac 12
	\end{array}  \Big| \frac{\cos(\theta) -\cos (u)}{2 \cos (\theta) }\right)    C_{2l+m}^{n+1}(\cos u) du. \nonumber
	\end{align}
\end{theorem}

\begin{proof}
	An integration by parts starting from \eqref{ii}, keeping in mind \eqref{dnGegen} 
	yields  
	\begin{align*}
	P_{\ell}^{(n,m)}(2t^2-1) 
	&=   \frac{ (-1)^m  n! d_{n,m}(\ell)}{2^m(m+n)! t^m} \int_0^{1} \frac{d^m}{dv^m}\left( (1-v^2)^{m-1/2} \right)  C_{2l+m}^{n+1} (tv)  dv,
	\end{align*}
	where $d_{n,m}(\ell)$ stands for the constant in \eqref{CSTlEM}.
	 Now, by Rodrigues formula for Jacobi polynomials, we have 
	\begin{align*} 
	\frac{d^m}{dv^m}\left( (1-v^2)^{m-1/2} \right)
	& = (-1)^m 2^m m! (1-v^2)^{-1/2} P_m^{(-1/2,-1/2)} (v)
	\\&= 
	  \frac{(-1)^m  (2m)!}{2^{m} } (1-v^2)^{-1/2}  {_2 F_1}\left( 
\begin{array}{c}	-m , m \\ \frac 12\end{array}  \Big| \frac{1-v}{2 }\right),
	\end{align*}
	it follows
	\begin{align*}
	P_{\ell}^{(n,m)}(2t^2-1) 
	&=  \frac{(2m)!n!   d_{n,m}(\ell)}{2^{2m}(m+n)! t^m}     
	\int_0^{1} (1-v^2)^{-1/2}  {_2 F_1}\left( 
	\begin{array}{c}	-m , m \\ \frac 12\end{array}  \Big| \frac{1-v}{2 }\right)  C_{2l+m}^{n+1} (tv)  dv.
	\end{align*}
	Finally, the change of variables $t=\cos(\theta)$ and $v=\cos(u)/\cos(\theta)$ completes the proof of Theorem \ref{corIntRepJPGauss}.
\end{proof}

\section{Application to Heat kernel on the  quantized Riemann sphere  $ S^2 $}

In the present section, we provide a concrete application of \eqref{IntRepJPol}. Indeed, we give the explicit integral representation for the heat kernel  $E_{\nu}(t,z,w) $  solving the following  Heat problem 
$$ \Delta_{\nu} E_{\nu}(t,z,z_0) =\frac{\partial}{\partial t} E_{\nu}(t,z,z_0) ; \quad, t>0 , \,z,z_0 \in S^2$$
and 
$$
\lim_{t\rightarrow 0} \int_{S^2 }E_{\nu}(t,z,w)f(w)d\mu_{\nu}(w)=f(z)  \in \C^{\infty}( S^2)
 $$
for $ \Delta_{\nu}$ in  \eqref{MagnLap}. 
The concrete  spectral analysis of the magnetic Laplacian $\Delta_{\nu}$ on $ S^2$ follows from the one elaborated by Peetre and Zhang in \cite{PeetreZhang1993} for 
$$    \widetilde{\Delta_{\nu}} =
-(1+|z|^2)^2 \frac{\partial^ 2}{\partial z \partial \overline{z}} +2\nu (1+|z|^2) \overline{z}\frac{\partial}{\partial\overline{z}} ,$$
 by observing that $\Delta_{\nu}$ and $\widetilde{\Delta_{\nu}}$ are unitary equivalent. In fact, for every sufficiently differential function 
 $$f \in L^2(S^2)=L^2\left( S^2, d\mu\right) ; \quad d\mu(z) := \frac{ dxdy}{\pi (1+|z|^2)^{2}} ,$$
  we have 
$$\Delta_{\nu} f = (1+|z|^2)^{-\nu}\left( \widetilde{\Delta_{\nu}} +\nu \right) \left( (1+|z|^2)^\nu f\right) .$$
Thus, the spectrum of $\Delta_{\nu}$ acting in the Hilbert space
$L^2(S^2)$ is purely discrete and consists of an infinite number of eigenvalues
$$  \lambda_{\nu, m} = =\nu + m(m+2\nu+1) ;
 \,  m=0,1,2, \cdots.
$$
Therefore, the spectral decomposition of the Hilbert space $ L^2(S^2)$ in terms of the eigenspaces $$  \mathcal{A}^{2,\nu}_\ell=\mathcal{A}^{2,\nu}_\ell(S^2) =\{ \phi:S^2 \rightarrow  \mathbb{C} \in L^2(S^2)  ; \quad    \Delta_{\nu} \phi= \lambda_{\nu, \ell}\phi      \}$$
reads 
$$  L^2(S^2) = \bigoplus_{\ell=0}^{+\infty} \mathcal{A}^{2,\nu}_\ell(S^2) .$$
Moreover, the $m$-th eigenspace $\mathcal{A}^{2,\nu}_\ell$
is a finite dimensional vector space  with dimension  $ 2\ell+2\nu+ 1 $.  
Moreover, the closed expression of the corresponding reproducing kernel
is given in \cite[Theorem 1, p. 231]{PeetreZhang1993}. It can be rewritten as
\begin{align} 
K^\nu_{m}(z,w)
&=\frac{(2\nu+2\ell+1)(1+z\overline{w})^{2\nu}}{(1+|z|^2)^{\nu} (1+|w|^2)^{\nu}}
{_2F_1}\left( \begin{array}{c} -\ell, \ell+2\nu+1 \\  1
\end{array} \Big | \sin^2(d(z,w)) \right), \nonumber
\\&=\frac{(2\nu+2\ell+1)(1+z\overline{w})^{2\nu}}{(1+|z|^2)^{\nu} (1+|w|^2)^{\nu}}
P^{(0,2\nu)}_\ell(\cos^2(2d(z,w))). \label{repKer2}
\end{align} 
where 
$$d(z,w)= \frac{ |1+z\overline{w}|}{(1+|z|)^2 (1+|w|)^2},$$
thanks to 
$${_2F_1}\left(\begin{array}{c}-m,1+\alpha+\beta+m\\\alpha+1
\end{array} \Big | \tfrac{1}{2}(1-z)\right)=\frac{m!}{(\alpha+1)_m} P_m^{(\alpha,\beta)}(z).$$

Accordingly, we can provide an expansion series 
of the  heat kernel  $E_{\nu}(t,z,z_0) $. 

\begin{proposition}\label{propHeatexp}
	The heat kernel $E_{\nu}(t,z,w)$  has the following asymptotic decomposition 
$$	E_{\nu}(t,z,w)=\frac{(1+z\overline{w})^{2\nu} e^{\nu t}}{(1+|z|^2)^{\nu} (1+|w|^2)^{\nu}} \sum_{\ell=0}^{+\infty} (2l+2\nu+1)e^{-l(l+2\nu+1)t} P_\ell^{(0,2\nu )}(\cos(2d(z,w))) .
$$
\end{proposition}

\begin{proof}
	The proof follows making use of the fact that for given self-adjoint operator  with eigenvalues $\lambda_j$ and the corresponding eigenfunctions $\{e_j\}$ is a complete orthonormal system, the heat kernel $E(t,z,z_0)$ of is given by 
	$$  E(t,z,z_0)= \sum_{k=0}^\infty e^{- \lambda_k t} e_k(z)\overline{e_k(z_0)} .$$
	See \cite{Davies} for example.    
	Therefore, the expansion in Proposition \ref{propHeatexp}  readily follows by means of the closed formula of $K^{\nu}_m$ given through \eqref{repKer2} since
		\begin{align*}
	E_{\nu}(t,z,z_0) &= \sum_{m=0}^{+\infty} 
	e^{- \lambda_{\nu,m} t} \left( \sum_{j=-m}^{m+2\nu} \frac{\phi^\nu_{m, j}(z) \overline{\phi^\nu_{m, j}(z_0)} }{\norm{\phi^\nu_{m, j}}^2_{L^2(S^2)}}\right) 
	\\&
	 = \sum_{m=0}^{+\infty} 
	e^{- \lambda_{\nu,m} t} K^{\nu}_m(z,z_0)
	.
	\end{align*}
\end{proof}

\begin{remark}\label{rem}
	By taking $\nu=0$, we recover the heat kernel associated to the Laplace--Beltrami operator $\frac{\partial^2}{\partial z \partial \bz  }$ on the Riemann sphere \cite{FisherJungsterWilliams1985}, 	
	$$ E_{0}(t;d)= \sum_{l=0}^{+\infty} (2l+1)e^{-l(l+1)t} P_l(cos(2d)) .$$
\end{remark}
	
	By means of the Dirichlet--Mehler integral representation for the Legendre polynomials \eqref{IntRepDirichlet--Mehler}, we can  rewrite $ E_{0}(t;d)$ in Remark \ref{rem}  in terms of the usual theta function
	$$	\theta_{2}(u)= \sum_{l=0}^{+\infty} e^{-l(l+1)t}\cos (2l +1)u$$
	as 
		\begin{align*}
		E_{0}(t,z,w) = \frac{2 }{\pi }  \int_{d}^{{\pi}/2} \frac{  \frac{d}{du} \Big( \theta_{2,0}(t,u) \Big)} { \sqrt  {\cos^2(d)-\cos^2(u)}}  du.
	\end{align*}
More generally, we prove the following. 

\begin{theorem} The explicit real integral representation of the Heat kernel  $E_{\nu}(t,z,w)$ for the invariant Laplacian $\Delta_\nu$ on the quantized  Riemann sphere  $ S^2 $ is given by
			\begin{align*}
E_{\nu}(t,z,w)& = \frac{2(1+z\overline{w})^{2\nu} e^{_\nu t}}{\pi (1+|z|^2)^{\nu} (1+|w|^2)^{\nu} \cos ^{2\nu}(d)}  \int_{d}^{{\pi}/2} \frac{  \frac{d}{du} \Big( \theta_{2,\nu}(t,u) \Big)} { \sqrt  {\cos^2(d)-\cos^2(u)}}
 \\& \qquad \qquad \times {_2 F_1}\left( 
\begin{array}{c}
-2\nu , 2\nu \\ \frac 12
\end{array}  \Big| \frac{\cos(d) -\cos (u)}{2 \cos(d) }\right)   du.
\end{align*}
	where $ d=d(z,w)$ and $ \theta_{2,\nu}(u)$  is given  by
	\begin{align}\label{theta2nu}
	\theta_{2,\nu}(u)= \sum_{l=0}^{+\infty} e^{-l(l+2\nu+1)t}\cos (2l+2\nu +1)u .
	\end{align}
\end{theorem}

\begin{proof}
	The closed integral representation of $	E_{\nu}(t,z,w)$ follows making use of  Proposition \ref{propHeatexp} as well as the integral representation of Jacobi polynomials given in  Theorem  \ref{corIntRepJPGauss}. Indeed, 
	\begin{align*}
		E_{\nu}(t,z,w)
			&= \frac{2(1+z\overline{w})^{2\nu} e^{\nu t}}{\pi (1+|z|^2)^{\nu} (1+|w|^2)^{\nu} \cos ^{2\nu}(d)} \int_{d}^{{\pi}/2} \frac{\sin(u)} { \sqrt  {\cos^2(d)-\cos^2(u)}}
\\&\qquad \times {_2 F_1}\left( 
\begin{array}{c}
-2\nu , 2\nu \\ \frac 12
\end{array}  \Big| \frac{\cos(d) -\cos (u)}{2 \cos(d) }\right)   R^\nu_\ell (u)  du,
		\end{align*}
		where we have set 
		$$ R^\nu_\ell (u) := \sum_{\ell=0}^{+\infty} (2l+2\nu+1)e^{-l(l+2\nu+1)t}  C_{2l+2\nu}^{1}(\cos u).$$
Finally, using \eqref{sinGrn}, we can rewrite $R^\nu_\ell (u)$ in terms of $\theta_{2,\nu}$ in \eqref{theta2nu} as 
		$$ R^\nu_\ell (u)  =\frac {1}{\sin(u)} \frac{d}{du}  \Big( \theta_{2,\nu}(t,u) \Big) .$$
\end{proof}

\section{A new proof of Dijksama-Koornwinder integral representation}

	The integral representation  \eqref{IntRepJPol}, for Jacobi polynomials in terms of ultraspherical polynomials, appears a specific case of
\begin{eqnarray}\label{IntRepProdDK1971} 
&&	P_n^{(\alpha, \beta)}(1-2t^2) 
P_n^{(\alpha, \beta)}(1-2s^2)
= 
\frac{ \Gamma( \alpha+\beta+1) \Gamma(n+\alpha+1)\Gamma(n+\beta+1)}{\pi n!\Gamma( n+\alpha+\beta+1)\Gamma\left( \alpha+\frac 12\right)  \Gamma\left( \beta+\frac 12\right) } 
\\& &\times 
\int_{-1}^1 \int_{-1}^1 C_{2n}^{(\alpha+\beta+1)}\left( stu+v\sqrt{(1-t^2)(1-s^2)}\right)  (1-u^2)^{\alpha -\frac 12}(1-v^2)^{\beta -\frac 12} dudv \nonumber
\end{eqnarray}
valid for two fixed complex numbers $\alpha,\beta$ such that $2\Re e (\alpha) > - 1$ and $2\Re e (\beta) > - 1$.
The proof of \eqref{IntRepProdDK1971} requires  special geometrical characterization of $P_n^{(\alpha, \beta)}$ (as invariant spherical harmonics under some orthogonal transformations in high dimensions) and the Laplace's integral  representation obtained by Braaksma and Meulenbeld in \cite{BraaksmaMeulenbeld1968}.

The proof we propose for \eqref{IntRepJPol} makes use of  Dirichlet--Mehler integral \eqref{IntRepDirichlet--Mehler} for the Legendre polynomials and is contained in the following fundamental and elementary lemmas. In fact, we need only to establish \eqref{IntRepJPol} for nonnegative integers $\alpha=n$ and $\beta =m$. The result for arbitrary complex numbers $\alpha,\beta$ such that $2\Re(\alpha)>-1$ and $2\Re(\beta)>-1$ follows by analytic continuation.


\begin{proof}[Proof of \eqref{ii}]
	We begin by noting that for every real $a$ such that $a\ne1$, we have the identity
	\begin{align}\label{fdleid}
	(1-v^2)^a  \frac{\partial}{4t\partial t} \left(   C_{k}^{(\lambda)}(tv) \right) &=
	-\frac{\lambda}{4(a+1)t^2} \frac{\partial}{\partial v} \left( (1-v^2)^{a+1}
	C_{k-1}^{(\lambda+1)}(tv) \right)\\ & \qquad +
	\frac{\lambda(\lambda+1)}{2(a+1)} (1-v^2)^{a+1} C_{k-2}^{(\lambda+2)}(tv) .\nonumber
	\end{align}
	This is easy to handle by observing that 
	$$\frac{\partial}{\partial t} \left( C_{k}^{(\lambda)}(tv) \right) = \frac{v}{t} \frac{\partial}{\partial v} \left( C_{k}^{(\lambda)}(tv) \right)
	$$ 
	and next using the well-established facts $f'g'= (fg')'-fg"$ and $ \frac{d}{dx}C_{\ell+1}^{(\lambda) } =  2 \lambda C_{\ell}^{(\lambda+1) }$. 
	Therefore, we get 
	\begin{align}\label{fdleid2}
	\int_0^{1} (1-v^2)^a  \frac{\partial}{4t\partial t} \left(   C_{2\ell}^{(\lambda)}(tv) \right) dv = 
	\frac{\lambda(\lambda+1)}{2(a+1)} \int_0^{1} (1-v^2)^{a+1} C_{2\ell-2}^{(\lambda+2)}(tv) dv
	\end{align}
	for $C_{2\ell-1}^{(\lambda)}(0)=0$.
	More generally, an inductive reasoning making use of \eqref{fdleid2} gives rise to 
	$$\int_0^{1} (1-v^2)^a \left(  \frac{\partial}{4t\partial t} \right) ^m \left(   C_{2\ell}^{(\lambda)}(tv) \right) dv = d_{a,\lambda}(n) \int_0^{1} (1-v^2)^{a+m} C_{2\ell-2m}^{(\lambda+2m)}(tv) dv,
	$$ 
	for some constant $d_{a,\lambda}(n)$ depending only in $a$, $\lambda$ and $n$.
	Now, by taking $a=-1/2$ and $\lambda=n-m+1$ with $n\geq m$, and using the explicit expression of the $m$-th derivative formula for the Jacobi polynomials \cite[p. 260]{Rainville71}
	$$ \left(  \frac{d}{dx} \right) ^m P_{\ell +m}^{(n,0)}(x)=\frac{(\ell+n+2m)!}{2^m (n+m+\ell)!} P_{\ell}^{(n+m,m)}(x) ,$$
	as well as Lemma \ref{lemIntJacGegn0}, we get 
	\begin{align*}\label{mthder}
	P_{\ell}^{(n,m)}(2t^2-1) 
	&= \frac{2^m (\ell+n)!}{(\ell+n+m)!} \left( \frac{d}{4tdt}\right)^m P_{\ell+m}^{(n-m,0)}(2t^2-1) 
	\\&\stackrel{\eqref{RelJacGegen}}{=} 
	\widetilde{s_{n,m}(\ell)}  \int_0^{1} \left( \frac{\partial}{4t\partial t}\right)^m \left( (1-v^2)^{-1/2}  C_{2(\ell+m)}^{(n-m+1)}(tv) \right)   dv
	\\&= s_{n,m}(\ell)  \int_0^{1} (1-v^2)^{m-1/2} C_{2\ell}^{(n+m+1)}(tv) dv
	\end{align*}
	for every nonnegative integers $n\geq m$. 
	The involved constant $s_{n,m}(\ell) $ is given by 
	$$ s_{n,m}(\ell) := \frac{2(n+m)! (m+\ell)!}{\sqrt{\pi} (n+m+\ell)!  \Gamma\left(m+\frac 12 \right) }$$
	and can be verified by taking $t=0$, keeping in mind the specific values of  
	$$C_{2\ell}^{(\lambda)}(0)=(-1)^{\ell} \frac{\Gamma(\lambda+\ell)}{\ell!\Gamma(\lambda)},$$
	$$\int_0^1(1-v^2)^{\alpha-1/2} dv =\frac{\sqrt{\pi} \Gamma\left( \alpha +\frac 12\right) }{\Gamma(\alpha+1)},$$
	and 
	$$ P_{\ell}^{(\alpha,\beta)}(-1)=(-1)^\ell \frac{\Gamma(\beta+\ell+1)}{\ell!\Gamma(\beta+1)}.$$
	This proves \eqref{ii}.
\end{proof}

\begin{remark}
	Using the symmetry relation \cite[Eq. (2.13), p. 8]{Askey1975}
	$$P_\ell^{(\alpha,\beta)}(-x) =(-1)^nP_\ell^{(\beta,\alpha)}(x),$$
	we recover \eqref{IntRepJPol}.
\end{remark}

\begin{remark}
	One recovers Mehler's form of Dirichlet's integral \eqref{IntRepDirichlet--Mehler} for 
	Legendre polynomials
	by taking $ \alpha =\beta=0$ in \eqref{IntRepJPol} and making  specific change of variables.
\end{remark}

 \end{document}